\documentclass[10pt,fullpage]{article}
\usepackage{color}
\usepackage{amssymb}
\usepackage{amsmath}
\usepackage{amsthm}
\usepackage{verbatim}

\newtheorem{theorem}{Theorem}

\newtheorem{proposition}{Proposition}
\newtheorem{conjecture}{Conjecture}

\usepackage{array}

\newcommand{\prob}{\mathrm{Prob}}

\author{
Tarik Kaced\thanks{LIF de Marseille, Univ. Aix--Marseille}
\and  Andrei Romashchenko\thanks{LIF de Marseille, CNRS \& Univ. Aix--Marseille on leave from IITP of RAS, Moscow}
}
\title{On essentially conditional information inequalities}

\begin{document}

\maketitle

\begin{abstract}

In 1997, Z.~Zhang and R.W.~Yeung found the first example of a conditional
information inequality in four variables that is not ``Shannon-type''. This
linear inequality for entropies is called conditional (or constraint) since it
holds only under condition that some linear equations are satisfied for the
involved entropies. Later, the same authors and other researchers discovered
several unconditional information inequalities that do not follow from
Shannon's inequalities for entropy.

 In this paper we show that some non Shannon-type 
conditional inequalities are ``essentially'' conditional, i.e., they cannot be extended 
to any unconditional inequality. We prove one new essentially conditional information
inequality for Shannon's entropy and  discuss conditional information
inequalities for  Kolmogorov complexity.

\end{abstract}

\section{Introduction}

Let $(X_1,\ldots,X_n)$ be jointly distributed random variables on a finite
domain.  For this collection of random variables there are $2^n-1$  non-empty
subsets and for each subset we have a value of Shannon's entropy. We call this
family of entropies the \emph{entropy profile} of the distribution
$(X_1,\ldots,X_n)$.  Thus, to every $n$-tuple of jointly distributed random
variables there corresponds its entropy profile which is a vector of values in
$\mathbb{R}^{2^n-1}$. We say that a point in $\mathbb{R}^{2^n-1}$ is
\emph{entropic} if it is a vector of entropies for some distribution.

All entropic points satisfy different \emph{information inequalities} that
characterize the range of all entropies for $X_i$. The most known and
understood are so-called \emph{Shannon-type} inequalities, i.e., linear
combinations of basic inequalities of type $I(U:V|W)\ge0$, where $U,V,W$ are
any (possibly empty) subsets of the given family of random variables.

In 1998 Z.~Zhang and R.W.~Yeung proved the first example of  an unconditional
\emph{non Shannon-type} information inequality, 
which was a linear inequality for  entropies of $(X_1,X_2,X_3,X_4)$ that cannot be represented
as a combination of basic inequalities \cite{zy98}. Since this seminal paper of Zhang and Yeung was published, 
many
(in fact, \emph{infinitely many}) non Shannon-type linear information inequalities were proven, see, e.g., 
\cite{mmrv,matus,matus-inf,six-inequalities,projections}. These new inequalities were applied  
in problems of network coding \cite{network}, secret sharing \cite{secret-sharing}, 
%in algorithmic information theory \cite{extracting}, 
etc. 
However, these inequalities and their `physical meaning'
are still not very well understood.

In this paper we discuss \emph{conditional} (constraint) information inequalities. That is, we are
interested in linear information inequalities that are true only given some linear constraint for entropies.
Trivial examples of conditional inequalities can be easily derived from (unconditional) basic inequalities,
e.g., if $H(X_1)=0$ then $H(X_1,X_2)\le H(X_2)$. However, some conditional inequalities cannot be obtained
as a corollary of Shannon-type inequalities. The first example of  a nontrivial conditional inequality
was proven in \cite{zy97} (even before the first example of an unconditional non Shannon-type
inequality):
\begin{eqnarray} \label{cond-ineq-zy97}
\begin{split}
\lefteqn{ \mbox{if }I(A:B)=I(A:B|C)=0, \mbox{ then}}\\
&&I(C:D)\le I(C:D|A)+I(C:D|B)
\end{split}
\end{eqnarray}
Another conditional inequality
\begin{eqnarray} \label{cond-ineq-matus}
\begin{split}
\lefteqn{ \mbox{if }I(A:B|C)=I(B:D|C)=0, \mbox{ then}}\\
&&I(C:D)\le I(C:D|A)+I(C:D|B)+ I(A:B) 
\end{split}
\end{eqnarray}
was proven by F.~Mat\'u\v{s} in~\cite{matus99}.

In~\cite{mmrv} it was conjectured that~(\ref{cond-ineq-zy97}) can be extended to some
unconditional inequality
\begin{eqnarray}\label{ineq-conjecture-mmrv}
\begin{split}
 \lefteqn{  I(C:D)\le I(C:D|A)+I(C:D|B) + }\\
 &&    \rule{40pt}{0pt}  + \kappa ( I(A:B)+I(A:B|C) )
%  I(C:D)\le I(C:D|A)+I(C:D|B) + \kappa ( I(A:B)+I(A:B|C) )
 \end{split}
\end{eqnarray}
(for some constant $\kappa>0$). In this paper we prove that this conjecture is \emph{wrong}:
for any coefficient $\kappa$, inequality (\ref{ineq-conjecture-mmrv}) is not true for some distributions.
So, inequality~(\ref{cond-ineq-zy97}) is ``essentially conditional''; it cannot be extended  to an unconditional
information inequality. A similar statement can be proven
for~(\ref{cond-ineq-matus}).

In this paper  we also prove one new conditional linear inequality that cannot be extended to any unconditional inequality.  So, now we have three examples of essentially conditional linear information inequality.  

It should be noticed that  these conditional information inequalities are proven for the set of \emph{entropic} points. 
(These three inequalities involve $4$-tuples of random variables, so technically they are some statements about the set of the entropic points in $\mathbb{R}^{15}$.) But  it is not  know whether they hold for the \emph{almost entropic} points (i.e., for the points $\mathbf{x}\in\mathbb{R}^{15}$ such that for every $\varepsilon>0$ there exists an entropic point $\mathbf{y}\in\mathbb{R}^{15}$ at the distance less than $\varepsilon$ from $\mathbf{x}$).  
In fact, the set of the almost entropic points is a nice  and interesting object to study.  In some sense, the almost entropic points make a more natural object than the entropic points; 
 e.g., 
for every $n$ the set of all almost entropic points for $n$-tuples of random variables is a closed convex cone (while 
 for  $n>2$ the set of all entropic points is not closed and not a cone). 
 We recall that some \emph{piecewise linear} conditional information inequality proven in~\cite{matus-piecewise} holds only for the entropic but not for the almost entropic points. 
 So there is an interesting open question: Do inequalities  (\ref{cond-ineq-zy97}), (\ref{cond-ineq-matus}), and the inequality from Theorem~\ref{th-cond-ineq}
hold for the  \emph{almost entropic} points?

It is known that the class of unconditional linear information inequalities are the same
for Shannon's entropy and for Kolmogorov complexity.
The situation with conditional inequalities is more complicated:
the known technique used to prove constraint information inequalities for Shannon's
entropy cannot be directly adapted for Kolmogorov complexity. In fact, it is not even
clear how to formulate Kolmogorov's version of constraint inequalities. However, we prove  
for Kolmogorov complexities some  counterpart of inequality~(\ref{cond-ineq-zy97}); 
this inequality holds only for some special tuples of words.

The paper is organized as follows.
In Section~\ref{sec3} we use the technique from \cite{zy97} and prove one new conditional 
information inequality. 
In Section~\ref{sec4} we prove that this new inequality as well as~(\ref{cond-ineq-zy97}) 
and~(\ref{cond-ineq-matus})
cannot be extended to any unconditional inequalities. 
In Section~\ref{sec-kolmogorov}  we prove some version  of conditional
inequality for Kolmogorov complexities.

\subsection{Corrected errors}

Some errors were found in the previous versions of the paper:

\begin{itemize}
\item The statement of Theorem~4(b) in  \texttt{arXiv:1103.2545v1} was wrong.
\item The proof of Theorem~4 in  \texttt{arXiv:1103.2545v2} and  \texttt{arXiv:103.2545v3}
and in the proceedings of ISIT-2011 was wrong. Note that the statement of this theorem
(it claims that the  cone of asymptotically entropic points for 4 random variables is not polyhedral) 
is true, see~\cite{matus-inf}.
However, the ``new proof'' of this result suggested in our paper was not valid since the conditional inequalities
under consideration are proven only for the entropic but not for the almost entropic points. We thank F.~Mat\'{u}\v{s},
who pointed out this mistake.
\end{itemize}

\section{Nontrivial conditional information inequalities} \label{sec3}

The very first example of an inequality that does  not follow
from basic (Shannon type) inequalities was the following
result of Z.~Zhang and R.~W.~Yeung:
\begin{theorem}[Zhang--Yeung, \cite{zy97}]\label{th-zy97}
For all random variables $A,B,C,D$, if $ I(A:B|C) = I(A:B) = 0$ then
$$
I(C:D)\le I(C:D|A) + I(C:D|B).
$$
\end{theorem}
With the same technique F.~Mat\'u\v{s} proved another conditional inequality~(\ref{cond-ineq-matus}), see~\cite{matus99}. Using a  similar method,  
we prove one new conditional inequality:
\begin{theorem}\label{th-cond-ineq}
For all random variables $A,B,C,D$ if $$H(C|A,B) = I(A:B|C) = 0,$$  then
$
I(C:D)\le I(C:D|A) + I(C:D|B)+ I(A:B).
$
\end{theorem}
\begin{proof}  The argument consists of two steps: \emph{enforcing conditional independence}
and \emph{elimination of conditional entropy}.
Let us have a joint distribution of random variables $A,B,C,D$. The first trick of the argument
is a special transformation of this  distribution:  we keep 
the same distribution of the triples $(A,C,D)$ and $(B,C,D)$ but make $A$ and $B$ independent conditional
on $(C,D)$. Intuitively it means that we first choose at random (using the old distribution) values of $C$ and $D$;
then given fixed values of $C,D$ we independently choose at random $A$ and $B$ (the conditional 
distributions of $A$ given $(C,D)$ and $B$ given $(C,D)$ are the same as in the original distribution). 

More formally, we construct a new distribution $(\tilde A,\tilde B,\tilde C, \tilde D)$.
If $\prob[A=a, B=b, C=c, D=d]$ is the original distribution, then the new distribution  is
defined as follows:
\begin{eqnarray*}
 \lefteqn{ \prob [\tilde A=a, \tilde B=b, \tilde C=c, \tilde D=d] =} \\ 
 && \frac{\prob [A=a,C=c,D=d]\cdot \prob [B=b,C=c,D=d]}{\prob [C=c,D=d]}
\end{eqnarray*}
(with the convention $\frac00=0$ for all values $a,b,c,d$ of the four random variables).
%With some abuse of notation we denote the new random variables 
%by $\tilde A,\tilde B,\tilde C,\tilde D$. 
From the construction ($\tilde A$ and $\tilde B$ are independent given $\tilde C,\tilde D$) it follows that
 $$
 H(\tilde A,\tilde B,\tilde C,\tilde D) = H(\tilde C, \tilde D) + H(\tilde A|\tilde C,\tilde D) + H(\tilde B|\tilde C,\tilde D)
 $$
Since $(\tilde A,\tilde C,\tilde D)$ and $(\tilde B,\tilde C,\tilde D)$ have exactly the same distributions as the original 
$(A,C,D)$ and $(B,C,D)$ respectively, we have
 $$
 H(\tilde A,\tilde B,\tilde C,\tilde D) = H(C,D) + H(A|C,D) + H(B|C,D)
 $$
The same entropy can be bounded in another way:
$$
H(\tilde A,\tilde B,\tilde C,\tilde D) \le H(\tilde D) + H(\tilde A|\tilde D)+ H(\tilde B|\tilde D)+ H(\tilde C | \tilde A,\tilde B)
$$
Notice that the entropies $H(\tilde D)$, $H(\tilde A|\tilde D)$ and $H(\tilde B|\tilde D)$ are equal 
to $H(D)$, $H(A|D)$ and $H(B|D)$
respectively (we again use the fact that $\tilde A,\tilde D$ and $\tilde B,\tilde D$ have the same distributions as $A,D$
and $B,D$ respectively in the original distribution). Thus, we get
\begin{eqnarray*}
 \lefteqn{ H(C,D) + H(A|C,D) + H(B|C,D)\le }  \\
 & & H(D) + H(A|D)+ H(B|D)+ H(\tilde C | \tilde A,\tilde B)
%  H(C,D) + H(A|C,D) + H(B|C,D)\le   H(D) + H(A|D)+ H(B|D)+ H(C | \tilde A,\tilde B)
\end{eqnarray*}
It remains to estimate the value $H(\tilde C | \tilde A,\tilde B)$. We will show that it is zero (and this is the
second trick used in the argument).

Here we will use the two conditions of the theorem. We say that some values $a,c$ 
(respectively, $b,c$ or $a,b$)
are \emph{compatible} if in the original distribution these values can appear together, i.e., 
$\prob[A=a,C=c]>0$ (respectively, $\prob[B=b,C=c]>0$ or $\prob[A=a,B=b]>0$). Since $A$ and $B$ are independent
given $C$, if some values $a$ and $b$ (of $A$ and $B$) are compatible with the same value $c$ of $C$,
then these $a$ and $b$ are compatible with each other.

In the new distribution $(\tilde A,\tilde B,\tilde C,\tilde D)$ values of $\tilde A$ and $\tilde B$ 
are compatible with each other \emph{only if}
they are compatible with one and the same value of $\tilde C$; 
hence, these values must be also compatible with each other in the original
distribution $(A,B)$. Further, since $H(C|A,B)=0$, for each pair of compatible values of $A,B$ there exists 
only one value of $C$. Thus, for each pair of values $(\tilde A,\tilde B)$ 
with probability $1$ there exists only one value of $\tilde C$.
In a word, in the new distribution $H(\tilde C|\tilde A,\tilde B)=0$.  

Summarizing our arguments, we get 
\begin{eqnarray*}
 \lefteqn{  H(C,D) + H(A|C,D) + H(B|C,D) \le }\\
 && H(D) + H(A|D)+ H(B|D),
%   H(C,D) + H(A|C,D) + H(B|C,D) \le H(D) + H(A|D)+ H(B|D),
\end{eqnarray*}
which is equivalent to 
$$
 I(C:D)\le I(C:D|A) + I(C:D|B) + I(A:B).
$$
\end{proof}

%\begin{comment}
The proof of Theorem~\ref{th-cond-ineq} presented above is based implicitly on non-negativity of
the Kullbak--Leibler divergence. The same idea can be presented in a slightly different form,
with an explicit reference to the Kullbak--Leibler inequality. The argument is almost the same
as the  proof of the second part of  Proposition~2.1 in\cite{matus99}:
\begin{proof}[Second version of the proof of Theorem~\ref{th-cond-ineq}:]
Let $p[a,b,c,d]$ be a distribution of $(A,B,C,D)$ such that $H(C|AB) = I(A:B|C)=0$. 
With some abuse of notations for we denote projections of this distribution as 
$$
p[a,c,d]=\prob[A=a,C=c,D=d],\ 
p[a,d]=\prob[A=a,D=d],\mbox{ etc.}$$ 

We construct two new distributions, $\tilde p [a,  b, c, d] = \prob[\tilde A=a, \tilde B=b, \tilde C=c, \tilde D=d]$, and
$\hat p [a,  b, c, d] = \prob[\hat A=a, \hat B=b, \hat C=c, \hat D=d]$. We define them as follows:
$$
 \tilde p [a,  b, c, d] =
\frac{p[a,c,d]\cdot p[b,c,d]}{p[c,d]}
$$
and
$$
 \hat p [a, b, c, d] = \left\{
 \begin{array}{cl}
 \frac{p[a,d]\cdot p[b,d]}{p_{D}[d]}, & \mbox{if } p [a, b, c] >0,\\
 0,&  \mbox{otherwise}.
 \end{array}
\right.
$$
Since $I(A:B|C)=0$,  the condition $p [a, b, c] >0$ is true if and only if
$p[a, c] >0$ and $p [b, c] >0$.

Then we use non-negativity of the Kullback--Leibler divergence:
\begin{eqnarray*}
%\begin{split}
0\le D(\hat p || \tilde p) =
\sum \frac{p[a,c,d]\cdot p[b,c,d]}{p[c,d]} \cdot
\log \frac{p[a,c,d]\cdot p[b,c,d]\cdot p[d]}{p[c,d]\cdot p[a,d]\cdot p[b,d]}
%\end{split}
\end{eqnarray*}
(the sum over all values $a,b,c,d$ such that $p [a, b, c] >0$).
It follows immediately that
$$
0\le H(A,D)+ H(B,D) + H(C,D) -H(A,C,D)-H(B,C,D)- H(D).
$$
Now we add  the values $I(B:C|A)=H(A,C)+H(B,C)-H(A,B,C)-H(C)$ and
$H(C|A,B)=H(A,B,C)-H(A,B)$ to the right-hand side of the inequality
(both these values are equal to $0$ for our distribution). We obtain
$$
0\le I(C:D|A)+ I(C:D|B)  + I(A:B) - I(C:D),
$$
and we are done.
\end{proof}
%\end{comment}

\section{Conditional inequalities that cannot be extended to any unconditional inequalities}
\label{sec4}

In~\cite{mmrv} it was conjectured that the conditional inequality from Theorem~\ref{th-zy97}
is  a corollary of some \emph{unconditional} information inequality (which was not discovered yet): 
\begin{conjecture}[\cite{mmrv}]\label{conjecture-1}
For some constant $\kappa>0$ inequality~(\ref{ineq-conjecture-mmrv})
%\begin{eqnarray}\label{ineq-conjecture-1}
%\begin{split}
%% \lefteqn{  I(C:D)\le I(C:D|A)+I(C:D|B) + 
%% &&    \rule{40pt}{0pt}  + \kappa ( I(A:B)+I(A:B|C) )
%   I(C:D)\le I(C:D|A)+I(C:D|B) + \kappa ( I(A:B)+I(A:B|C) )
% \end{split}
%\end{eqnarray} 
is true for all random variables $A,B,C,D$. 
\end{conjecture}
\noindent
Obviously, if such an inequality could be proven, it would imply the statement of Theorem~\ref{th-zy97}.
Similar conjectures could be formulated for~(\ref{cond-ineq-matus}) and the conditional inequality from Theorem~\ref{th-cond-ineq}.
We prove that these conjectures are false, i.e.,  these three conditional inequalities cannot
be converted into unconditional  inequalities:

\begin{theorem}\label{th-not-unconditional}
(a) For any $\kappa$ the inequality~(\ref{ineq-conjecture-mmrv})
is not true for some distributions $(A,B,C,D)$.

(b) For any $\kappa$ the inequality
\begin{eqnarray}\label{false-conj-2}
\begin{split}
% \lefteqn{  I(C:D)\le I(C:D|A)+I(C:D|B) + 
% &&    \rule{40pt}{0pt}  + \kappa ( I(A:B)+I(A:B|C) )
\lefteqn{   I(C:D)\le I(C:D|A)+I(C:D|B) + I(A:B)+}\\
&&  \rule{80pt}{0pt} + \kappa ( I(A:B|C) + H(C|A,B))
 \end{split}
\end{eqnarray} 
is not true for some distributions $(A,B,C,D)$.

(c) For any $\kappa$ the inequality
\begin{eqnarray} \label{false-conj-3}
\begin{split}
\lefteqn{   I(C:D)\le I(C:D|A)+I(C:D|B) + I(A:B)+}\\
&&  \rule{80pt}{0pt} + \kappa ( I(A:B|C) + H(B:D|C))
\end{split}
\end{eqnarray}
is not true for some distributions $(A,B,C,D)$. Thus, (\ref{cond-ineq-matus}) cannot be extended
to an unconditional inequality.
\end{theorem}
\begin{proof}
(a) For all $\varepsilon\in[0,1]$ we us consider the following  joint distribution of binary variables 
$(A,B,C,D)$:
$$
\begin{array}{rcc}
 \prob[A=0,\ B=0,\ C=0,\  D=1] &=&  (1-\varepsilon)/4, \\
 \prob[A=0,\ B=1,\ C=0,\  D=0] &=& (1-\varepsilon)/4, \\
 \prob[A=1,\ B=0,\ C=0,\  D=1] &=& (1-\varepsilon)/4,  \\
 \prob[A=1,\ B=1,\ C=0,\  D=1] &=&(1-\varepsilon)/4,  \\
 \prob[A=1,\ B=0,\ C=1,\  D=1] &=& \varepsilon. 
\end{array}
$$
For each value of $A$ and for each values of $B$, the value
of at least one of variables $C,D$ is uniquely determined: if $A=0$ then $C=0$;
if $A=1$ then $D=1$; if $B=0$ then $D=1$; and if $B=1$ then $C=0$. Hence,
$I(C:D|A)=I(C:D|B)=0$.  Also it is easy to see that $I(A:B|C)=0$. Thus, if~(\ref{ineq-conjecture-mmrv})
is true, then
$
I(C:D) \le \kappa I(A:B).
$

Denote the right-hand and left-hand sides of this inequality by $L(\varepsilon)=I(C:D)$ and 
$R(\varepsilon)=\kappa I(A:B)$. Both  functions $L(\varepsilon)$ and $R(\varepsilon)$
are continuous, and $L(0)=R(0)=0$ (for $\varepsilon = 0$
both sides of the inequality are equal to $0$). However the asymptotics of 
$L(\varepsilon)$ and $R(\varepsilon)$ as $\varepsilon\to 0$ are different:
it is not hard to check that  
$L(\varepsilon) = \Theta(\varepsilon)$, but  $R(\varepsilon) = O(\varepsilon^2)$.
From~(\ref{ineq-conjecture-mmrv}) we have
$
\Theta(\varepsilon) \le O(\varepsilon^2),
$
which is a contradiction.

\medskip

(b) For every value of $\varepsilon\in[0,1]$ we consider 
the following  joint distribution of binary variables $(A,B,C,D)$:
$$
\begin{array}{rcc}
 \prob[A=1,\  B=1,\ C=0,\ D=0] &=& 1/2 -\varepsilon,\\
 \prob[A=0,\  B=1,\ C=1,\ D=0] &=& \varepsilon,\\
 \prob[A=1,\  B=0,\ C=1,\ D=0] &=& \varepsilon, \\
 \prob[A=0,\  B=0,\ C=1,\ D=1] &=& 1/2 -\varepsilon.
\end{array}
$$
The argument is similar to~the proof if~(a).
First, it is not hard to check that $I(C:D|A)=I(C:D|B)=H(C|AB)=0$
for every $\varepsilon$. Second, 

 $$
 \begin{array}{rcl}
 I(A:B)&=& 1 + (2-2/\ln 2)\varepsilon + 2\varepsilon\log \varepsilon + O(\varepsilon^2),\\
 I(C:D)&=& 1 +(4-2/\ln 2)\varepsilon + 2\varepsilon\log \varepsilon  + O(\varepsilon^2),
 \end{array}
 $$
so $I(A:B)$ and $I(C:D)$ both tend to $1$ as $\varepsilon\to 0$, but their asymptotics are different. 
Similarly, 
 $$
 I(A:B|C) = O(\varepsilon^2).
$$
It follows from~(\ref{false-conj-2}) that
$$
 2\varepsilon +O(\varepsilon^2) \le 
  O(\varepsilon^2) +  O(\kappa \varepsilon^2),
$$
and with any $\kappa$ we get a contradiction for small enough $\varepsilon$.

(c) For the sake of contradiction we consider the following  joint distribution of binary variables $(A,B,C,D)$
for every value of $\varepsilon\in[0,1]$:
$$
\begin{array}{rcc}
 \prob[A=0,\ B=0,\ C=0,\  D=0] &=& 3\varepsilon,\\
 \prob[A=1,\ B=1,\ C=0,\  D=0] &=& 1/3-\varepsilon,\\
 \prob[A=1,\ B=0,\ C=1,\  D=0] &=& 1/3-\varepsilon, \\
 \prob[A=0,\ B=1,\ C=0,\  D=1] &=& 1/3 -\varepsilon.
\end{array}
$$
We substitute this distribution in~(\ref{false-conj-3}) and obtain
$$
I_0+O(\varepsilon) \le I_0+ 3\varepsilon \log \varepsilon + O(\varepsilon)+ O(\kappa \varepsilon),
$$
where $I_0$ is the mutual information between $C$ and $D$ for $\varepsilon=0$ (which is 
equal to the mutual information between $A$ and $B$ for $\varepsilon=0$).
We get a contradiction as $\varepsilon\to 0$ .
\end{proof}

\section{Constraint inequality for Kolmogorov complexity} \label{sec-kolmogorov}

Kolmogorov complexity of a finite binary string $X$ is defined
as the length of the shortest program that generates $X$; 
similarly, Kolmogorov complexity of a string $X$ given 
another string $Y$ is defined as the length of the shortest 
program that generates $X$ given $Y$ as an input. More
formally, for any programming language $L$,
Kolmogorov complexity $K_L(X|Y)$ is defined as
$$
 K_L(X|Y) = \min\{ |p| \ :\ \mbox{program $p$ prints $X$ on input $Y$} \},
$$
and unconditional complexity $K_L(X)$ is defined as complexity of $X$
given the empty $Y$.
The basic fact of Kolmogorov complexity theory is
the invariance theorem: there exists a \emph{universal}
programming language $U$ such that for any other language
$L$ we have
$
K_U(X|Y) \le K_L(X|Y)+ O(1)
$
(the $O(1)$ depends on $L$ but not on $X$ and $Y$).
We fix such a universal language $U$; in what
follows we omit the subscript $U$ and denote Kolmogorov
complexity by $K(X)$, $K(X|Y)$. We refer the reader to an excellent 
book \cite{lv} for a survey of properties of Kolmogorov complexity.

Kolmogorov complexity was introduced in \cite{kol} as an
algorithmic version of measure of information in an individual
object. In some sense, properties of Kolmogorov complexity
are quite similar to properties Shannon's entropy. For example,
for the property of Shannon's entropy $H(A,B) =  H(A)+ H(B|A)$
there is a Kolmogorov's counterpart
 \begin{eqnarray}
 K(A,B) = K(A) + K(B|A)+O(\log K(A,B))\label{kl}
 \end{eqnarray}
(the Kolmogorov--Levin theorem, \cite{zl}). This result justifies 
the definition of the mutual information, which is an algorithmic
version of the standard Shannon's definition:
the mutual information is defined as
$
 I(A:B) := K(A) + K(B) - K(A,B),
$
and the conditional mutual information is defined as
$$ 
 I(A:B|C) := K(A,C) + K(B,C) - K(A,B,C) - K(C).
$$
From the Kolmogorov--Levin theorem it follows that $I(A:B)$ is equal
to $K(A)-K(A|B)$, and the conditional mutual
information $I(A:B|C)$ is equal to $K(A|C)-K(A|B,C)$
(all these equations hold only up to logarithmic terms).

In fact, we have a much more deep and general parallel between Shannon's and Kolmogorov's
information theories;
for  every linear inequality for Shannon's entropy there exists 
a Kolmogorov's counterpart: 
\begin{theorem}[\cite{hrsv}]\label{th-hrsv}
 For each family of coefficients $\{\lambda_W\}$ the inequality
  $$
  \sum\limits_{i}\lambda_i H(\alpha_i) +\sum\limits_{i<j}\lambda_{ij} H(\alpha_i,\alpha_j)+\ldots \ge 0
  $$
is true for every distribution $\{\alpha_i\}$  if and only if for some constant $C$
the inequality
  $$
  \sum\limits_{i}\lambda_i K(a_i) +\sum\limits_{i<j}\lambda_{ij} K(a_i,a_j)+\ldots C\log N\ge 0
  $$
is true for all tuples of strings $\{a_i\}$, $N=K(a_1,a_2,\ldots)$ ($C$ does not depend on $a_i$).
 \end{theorem}
Thus, the class of unconditional inequalities valid for Shannon's entropy coincides
with the class of (unconditional) inequalities valid for Kolmogorov complexity. 
What about conditional inequalities?

In the framework of Kolmogorov complexity we cannot say that some information quantity
\emph{exactly} equals zero. Indeed, even the definition of Kolmogorov complexity 
makes sense only up to an additive term that depends on the choice of the universal
programming language. Moreover, such a natural basic statement as the Kolmogorov--Levin
theorem~(\ref{kl}) holds only up to a logarithmic term. So, if we want to prove a sensible conditional
inequality for Kolmogorov complexity, the linear constraints must be formulated with some 
reasonable precision.
A natural version of Theorem~\ref{th-zy97} is the following conjecture:
\begin{conjecture}\label{conjecture-2} There exist functions $f(n)$ and $g(n)$ such that 
$f(n)=o(n)$ and $g(n)=o(n)$, and  for
all strings $A, B, C, D$ satisfying   
 $I(A:B|C) \le f(N),\   I(A:B) \le f(N)$
it holds
$
I(C:D)\le I(C:D|A) + I(C:D|B) + g(N)
$
(where $N=K(A,B,C,D)$).
\end{conjecture}
There is no hope to prove Conjecture~\ref{conjecture-2} with $f(n)$ and $g(n)$ of order $\Theta(\log n)$.
Indeed, using a counterexample
from the proof of Theorem~\ref{th-not-unconditional}(a), we can construct binary strings
$A, B, C, D$ such that the quantities $I(A:B|C)$, $I(A:B)$, $I(C:D|A)$, and $I(C:D|B)$
are bounded by $O(\log N)$, but $I(C:D)=\Omega(\sqrt{N\log N})$.
However, even if Conjecture~\ref{conjecture-2} is false in general,  similar conditional 
inequalities (even with logarithmic precision) can be true for some special tuples $A,B,C,D$. In what follows we
show how to prove such an inequality for one natural example of strings $A,B,C$ (and any $D$).

Let $\mathbb{F}_n$ be the finite field of $2^n$ elements.
We consider the affine plane over $\mathbb{F}_n$. Let $C$
be random line in this plane, and $A$ and $B$ be two points
incident to this line. To specify the triple $\langle A,B,C\rangle$
we need at most $4n+O(1)$ bits of information: a line in a plane
can be specified by two parameters in $\mathbb{F}_n$;  to
specify each point in a given line we need additional $n$ bits 
of information.
 
We take a triple of strings $\langle A,B,C\rangle$ as specified above 
with maximal possible  Kolmogorov complexity, i.e., such that 
 $
 K(A,B,C)=4n + O(1)
 $
(it follows from a simple counting argument that such a triple
exists; moreover, there are about $2^{4n+O(1)}$ such triples).
For these $A$, $B$ and $C$ we can easily estimate all
their Kolmogorov complexities:
$$
\begin{array}{l}
K(A),\ K(B),  \mbox{ and }K(C) \mbox{ are equal to } 2n+O(1),\\
K(A,C) = 3n+O(1),\ K(B,C)=3n+O(1),\\
H(A,B)=4n+O(1).
%I(A:B) = O(\log n),\ I(A:B|C) = O(\log n).
\end{array}
$$

For this triple of strings  the quantities $I(A:B)$ and $I(A:B|C)$
are negligible (logarithmic). This condition is very
similar to the condition on random variables
$A,B,C$ in Theorem~\ref{th-zy97}. So, it is 
not very surprising that Kolmogorov's counterpart of 
Theorem~\ref{th-zy97} holds for these strings:

\begin{proposition}\label{prop-kolmogorov}
For the strings $A,B,C$ defined above and for \emph{all} 
strings $D$ we have
 $$
 I(C:D) \le I(C:D|A) + I(C:D|B) + O(\log N),
 $$
where $N=K(A,B,C,D)$. 
\end{proposition}
This statement can be proven by an 
argument similar to the proof of Theorem~\ref{th-cond-ineq}.
Let us explain this argument in full detail.
\begin{proof} 
We may identify $C$ with a linear function
$c_1x+c_2$ over $\mathbb{F}_n$, where $c_1$ and $c_2$
are elements of the field (since Kolmogorov complexity of $C$
is large, it cannot be a \emph{vertical} line on the plane). Further,
the points $A$ and $B$ in this line can be represented as pairs 
$\langle a_1,a_2\rangle$ and $\langle b_1,b_2\rangle$ such that
 $$
 c_1 \cdot a_1 +c_2=a_2\mbox{ and }
 c_1 \cdot b_1 +c_2=b_2
 $$
(here $a_i$ and $b_i$ are also elements of $\mathbb{F}_n$).
By assumption, complexity of the pair $(A,B)$ is close to $4n$.
It means that $A\not= B$; hence, $a_1\not= b_1$. Let $i$ be 
one of indexes such that the $i$th bits of $a_1$ and $b_1$ 
are different. W.l.o.g. we assume that the $i$th bit in $a_1$ is
equal to $0$ and the $i$th bit in $b_1$ is equal to $1$.

Now we split the affine plane over $\mathbb{F}_n$ into two 
halves: $P_0$ will consist of all points $(x,y)$ such 
that the $i$th bit of $x$ is $0$, and  $P_1$ will consist of the
points $(x,y)$ such  that the $i$th bit of $x$ is $1$. So, point
$A=(a_1,a_2)$ belongs to $P_0$, and $B=(b_1,b_2)$ belongs
to $P_1$.

Now we are going to variate the points $A$ and $B$: we will
substitute $A$ and $B$ by their `clones' $A'$ and $B'$ so that
the triples $\langle A', B' ,C\rangle$ remain ``similar'' to
the initial one $\langle A, B,C\rangle$. More precisely, we say
that $A'$ is a \emph{clone} of $A$ if 
\begin{itemize}
\item $A'=(a'_1,a'_2)$ is a point in line $C$, and $A'\in P_0$
(i.e., $c_1 \cdot a'_1 +c_2=a'_2$, and the $i$th bit of $a'_1$
is equal to $0$);
\item complexities $K(A')$, $K(A',C)$, $K(A',D)$, and $K(A',C,D)$
are equal (up to an additive term $O(\log N)$) to the corresponding
complexities  $K(A)$, $K(A,C)$, $K(A,D)$, and $K(A,C,D)$.
\end{itemize}
Similarly, we say that $B'$ is a \emph{clone} of $B$ if 
\begin{itemize}
\item $B'=(b'_1,b'_2)$ is a point in line $C$, and $B'\in P_1$, and
\item complexities $K(B')$, $K(B',C)$, $K(B',D)$, and $K(B',C,D)$
are equal (up to an additive term $O(\log N)$) to the corresponding
complexities  $K(B)$, $K(B,C)$, $K(B,D)$, and $K(B,C,D)$.
\end{itemize}
From a simple counting argument it follows that there exist
$2^{K(A|C,D)-O(\log N)}$ different clones of $A$ and $2^{K(B|C,D)-O(\log N)}$
clones of $B$ (see, e.g.,  \cite[Lemma~2]{hrsv}  or \cite[Lemmas~1--2]{muchnik}).

Let us take a pair of clones $A'$ and $B'$ with maximal complexity
given $(C,D)$. Then 
\begin{eqnarray*}
 \lefteqn{  K(A',B',C,D) =  } \\
 &&\lefteqn{ K(C,D) + K(A'|CD)+ K(B'|CD) + O(\log N)=}\\
 && K(C,D) + K(A|C,D) + K(B|C,D) + O(\log N)
\end{eqnarray*}
On the other hand,
\begin{eqnarray*}
 \lefteqn{  K(A',B',C,D)  \le K(D) +}\\
 &&{ K(A'|D)+ K(B'|D) + K(C|A',B')+O(\log N)}
\end{eqnarray*}
By definition of clones, complexities $K(A'|D)$ and $K(B'|D)$ are equal (up to $O(\log N)$ term)
to $K(A|D)$ and $K(B|D)$ respectively. Since $A'$ and $B'$ belong to $P_0$ and $P_1$
respectively, they cannot be equal to each other. Hence, $A'$ and $B'$ uniquely
determine line $C$. So, we get 
\begin{eqnarray*}
 \lefteqn{ K(C,D) + K(A|CD)+ K(B|CD) \le }\\
 &&K(D) + K(A|D) + K(B|D) + O(\log N),
\end{eqnarray*}
which is equivalent (by the Kolmogorov--Levin theorem) to
\begin{eqnarray*}
 I(C:D) \le I(C:D|A) + I(C:D|B) + O(\log N).
\end{eqnarray*}
\end{proof}

\section{ Acknowledgement }

This work was partially supported  by EMC ANR-09-BLAN-0164-01 and NAFIT ANR-08-EMER-008-01 grants. 

We thank anonymous referees 
%\begin{comment}
of ISIT~2011
%\end{comment}
for useful comments that helped us to substantially rework the original manuscript. We are 
especially grateful to F.~Mat\'{u}\v{s}, who pointed out a mistake in the previous version if this paper.

\end{document}